\documentclass[12pt]{article}

\usepackage{amsmath,amssymb,amsfonts}
\usepackage{amsthm}
\usepackage{graphicx}
\usepackage{booktabs}
\usepackage{multirow}
\usepackage{array}
\usepackage{algorithm}
\usepackage{algorithmic}
\usepackage{float}
\usepackage{url}
\usepackage{geometry}
\usepackage{hyperref}
\usepackage{adjustbox}
\usepackage{placeins}

\newtheorem{theorem}{Theorem}
\newtheorem{lemma}{Lemma}
\newtheorem{corollary}{Corollary}
\newtheorem{definition}{Definition}

\newcommand{\Z}{\mathbb{Z}}
\newcommand{\Gk}{G_k}

\newcommand{\dist}{\operatorname{dist}}
\newcommand{\sgn}{\operatorname{sgn}}
\newcommand{\ecc}{\operatorname{ecc}}

\title{Multi-Orientation Edge-Minimum Repair for Non-Redundant Fault-Tolerant Broadcasting in Dense Gaussian Networks}

\author{Bader A. Albader\\
\small Department of Computer Science, Faculty of Science, Kuwait University, Kuwait\\
\small \texttt{albader@cs.ku.edu.kw}}

\date{}

\begin{document}

\maketitle

\begin{abstract}
Dense Gaussian networks are degree-four algebraic interconnection networks with compact diameter and simple modular routing. This paper studies non-redundant one-to-all broadcast repair in the dense Gaussian network generated by $\alpha=k+(k+1)i$. We propose multi-orientation edge-minimum repair (MOEM), which evaluates a constant-size family of Gaussian broadcast-tree orientations, selects a fault-aware orientation, contracts the fault-pruned tree into healthy components, and reconnects those components using external component-crossing repair edges. The resulting structure is a rooted spanning tree of the healthy subgraph, so each healthy node receives the message exactly once and no faulty node is used. We prove that, for a chosen orientation with $c$ fault-pruned components and a connected healthy component graph, the repair step is non-redundant and uses the minimum possible number $c-1$ of external component-repair edges. We also prove that, for every one- or two-fault placement, the MOEM orientation family contains a repair with depth at most $k+2$. The depth proof combines a certificate framework, an explicit four-case off-axis analysis, and a five-component orthogonal-axis certificate. Exhaustive validation for $k=5,\ldots,10$ and large-scale validation through $k=200$ confirm the implementation and show that random two-fault repairs use approximately two external repair edges.
\end{abstract}

\noindent\textbf{Keywords:} Gaussian networks, Cayley graphs, fault-tolerant broadcasting, non-redundant communication, local repair, interconnection networks, Hamiltonian cycles.

\section{Introduction}

One-to-all broadcasting is a fundamental collective communication primitive in parallel and distributed systems. In a non-redundant broadcast, each healthy node receives the message exactly once. This is desirable because it avoids duplicate traffic, reduces contention, and gives a simple correctness condition: the transmission structure must be a rooted spanning tree of the healthy subgraph. The challenge is that a broadcast tree designed for a fault-free topology can be fragmented by a small number of faulty processors.

Algebraic interconnection networks provide attractive settings for structured communication algorithms. Cayley-graph models give symmetry and compact routing descriptions \cite{akers1989}; tori and related networks have long been used to study collective communication and fault tolerance \cite{leighton1992,dally2004,almohammad1999}. Gaussian networks are quotient networks over Gaussian integers. They provide degree-four planar algebraic topologies with strong routing, placement, and Hamiltonian properties \cite{flahive2010,flahive2013,albader2016}. The dense Gaussian family considered here has particularly simple integer labels and a natural diameter-level coordinate ball.

Let
\[
    N=k^2+(k+1)^2
\]
and let $\Gk$ be the dense Gaussian network generated by $\alpha=k+(k+1)i$. The node set is $\Z_N$ and the edges are
\[
    u\sim u\pm k,\qquad u\sim u\pm(k+1)\pmod N.
\]
The map
\begin{equation}
    \phi(x+yi)\equiv kx+(k+1)y\pmod N
    \label{eq:phi}
\end{equation}
identifies the Manhattan ball
\[
    B_k=\{(x,y)\in\Z^2: |x|+|y|\le k\}
\]
with the $N$ nodes of $\Gk$. The ball has exactly $2k^2+2k+1=N$ points. Hence, a source-centered Gaussian broadcast tree can reach all nodes in depth at most $k$ in the fault-free case.

The problem studied in this paper is the following.

\begin{definition}[Non-redundant local broadcast repair]
Given a source $s$ and a fault set $F\subseteq V(\Gk)$ with $s\notin F$, construct a rooted spanning tree of $\Gk-F$ rooted at $s$. The construction should use no faulty node, should give every healthy non-source node exactly one parent, should keep the number of repair edges small, and should keep the broadcast depth close to the fault-free depth $k$.
\end{definition}

A single fixed Gaussian broadcast tree is vulnerable to trunk faults. If a fault lies on a major branch, the tree may split into several components. A depth-safe repair can reconnect the healthy nodes by adding many candidate crossing edges and taking a breadth-first tree, but this may substantially increase the repair count. The central observation of this paper is that one should not repair a poor orientation when another valid orientation may be much less damaged.

This distinction between static placement and post-fault broadcast repair motivates the orientation-selection step in MOEM.

This paper proposes \emph{multi-orientation edge-minimum repair} (MOEM). MOEM constructs a small family of deterministic Gaussian broadcast-tree orientations. For each orientation, it removes the faults, forms the component graph of the resulting forest, and applies an edge-minimum component repair. The final output is the valid candidate with best lexicographic score, primarily minimizing depth and then repair count. In the final data, MOEM avoids fallback entirely for random trials with $k\ge25$.

The main contributions are as follows.
\begin{itemize}
    \item We formulate non-redundant local repair for dense Gaussian one-to-all broadcasting under one and two processor faults.
    \item We introduce MOEM: an eight-orientation Gaussian broadcast family followed by edge-minimum component repair.
    \item We prove that repairing a chosen fault-pruned orientation with $c$ components requires and attains exactly $c-1$ external component-crossing edges.
    \item We prove a $k+2$ repaired-depth theorem for every one- and two-fault placement using coordinate certificates, including an explicit O3 off-axis classification and an O6 orthogonal-axis certificate.
    \item We validate all one- and two-fault placements for $k=5$ through $10$ and large-scale cases through $k=200$, and we identify hard cases showing that the balanced-smaller-first orientation is indispensable within the tested family.
\end{itemize}

This paper focuses on the first nontrivial multi-fault regime, $|F|\le2$. This scope enables a complete coordinate proof of the $k+2$ depth law. The component-repair theorem applies to any selected orientation and any fault count with a connected healthy component graph, while the depth-certificate theorem is proved here for one and two faults; Section~\ref{sec:discussion} sketches the three-fault certificate classes suggested by the proof.

\section{Related Work}

Akers and Krishnamurthy gave a group-theoretic model for symmetric interconnection networks, providing a foundation for algebraic network design \cite{akers1989}. Classical texts cover meshes, tori, hypercubes, and collective communication patterns \cite{leighton1992,dally2004}. AlMohammad and Bose studied non-redundant fault-tolerant communication in toroidal networks \cite{almohammad1999}; their work is conceptually related because it preserves collective communication after faults, but it does not repair a pruned dense Gaussian broadcast tree under an external component-crossing edge metric.

Gaussian and Eisenstein--Jacobi networks were developed as algebraic interconnection topologies with compact routing and useful symmetry. Martinez et al. modeled toroidal networks using Gaussian integers and studied dense Gaussian networks as low-degree on-chip multiprocessor topologies \cite{martinez2008,martinez2006}; related Martinez-group work includes hierarchical Gaussian, EJ/hexagonal, and lattice-derived models \cite{moreto2006,martinez2008ej,camarero2013,camarero2015}. Flahive and Bose studied Gaussian and Eisenstein--Jacobi topology and resource placement \cite{flahive2010,flahive2013}. Dense Gaussian networks also have strong Hamiltonian, routing, and path properties \cite{albader2012,alsaleh2013,albader2016,monakhova2020}. Network-on-chip references motivate compact low-degree graphs \cite{pasricha2008,kliazovich2013,song2020}, but this paper does not claim chip-level latency, bandwidth, or power evaluation.

The present work differs from placement, routing, Hamiltonian-cycle construction, and precomputed tree-diversity mechanisms. MOEM starts with a specific diameter-level broadcast tree after the fault set is known, preserves its healthy components, and minimizes new external component-crossing repair edges. A global BFS rebuild may restore connectivity but can replace $\Theta(N)$ parent edges; precomputed tree-diversity methods optimize path diversity rather than post-fault local repair. Thus fixed-orientation local repair is the nearest same-objective comparator.

\section{Dense Gaussian Network Model}

Let $\Gk=(V,E)$ be the dense Gaussian network with $V=\Z_N$ and $N=k^2+(k+1)^2$. For $u,v\in V$,
\[
    (u,v)\in E \quad\Longleftrightarrow\quad v-u\equiv \pm k \text{ or } \pm(k+1) \pmod N.
\]
Edges of difference $\pm k$ are called $H_1$-type edges, and edges of difference $\pm(k+1)$ are called $H_2$-type edges. In the dense case, $\gcd(k,k+1)=1$, so each generator induces a Hamiltonian cycle.

\subsection{Coordinate Representation}

Equation~\eqref{eq:phi} maps the coordinate ball $B_k$ onto $\Z_N$. Since
\[
    |B_k|=1+2k(k+1)=k^2+(k+1)^2=N,
\]
this representation is bijective. Moving from $(x,y)$ to $(x+1,y)$ changes the label by $k$, and moving from $(x,y)$ to $(x,y+1)$ changes the label by $k+1$. Therefore unit Manhattan moves in $B_k$ correspond exactly to Gaussian graph edges.

\subsection{Fault-Free Oriented Broadcast Trees}

For a fixed orientation $\theta$, define a parent rule that maps each nonzero coordinate $(x,y)$ to a neighbor closer to the origin. The $x$-first tree, for example, uses
\begin{equation}
    p_x(x,y)=
    \begin{cases}
    (x-\sgn(x),y), & x\ne 0,\\
    (x,y-\sgn(y)), & x=0,\;y\ne0.
    \end{cases}
    \label{eq:xfirst}
\end{equation}
Applying $\phi$ to these parent relations gives a spanning tree of $\Gk$ rooted at the source. The depth of coordinate $(x,y)$ is $|x|+|y|\le k$.

\begin{lemma}
For any coordinate-reduction orientation that decreases $|x|+|y|$ by one at every non-source node, the induced parent relation is a spanning tree of $\Gk$ with depth at most $k$.
\end{lemma}

\begin{proof}
Each non-source node $(x,y)$ is assigned a parent obtained by changing one nonzero coordinate by one unit toward zero. Hence the parent is adjacent in $\Gk$ and the value $|x|+|y|$ decreases by one. Repeated parent application reaches $(0,0)$ after exactly $|x|+|y|$ steps, so cycles are impossible and all nodes are connected to the source. Since $|x|+|y|\le k$ for all coordinates in $B_k$, the depth is at most $k$.
\end{proof}

MOEM uses eight deterministic orientations: $x$-first, $y$-first, positive-preference variants, negative-preference variants, larger-coordinate-first, and smaller-coordinate-first. This family is intentionally small and constant-size, so orientation selection does not change the asymptotic complexity.

\section{Multi-Orientation Edge-Minimum Repair}

\subsection{Component Graph}

Let $T_\theta$ be the broadcast tree for orientation $\theta$. Removing the faulty vertices and their incident tree edges gives a forest
\[
    T_\theta-F=C_1\cup C_2\cup\cdots\cup C_c.
\]
Let $C_s$ be the component containing the source. The component graph $\mathcal{C}_\theta$ has one vertex for each component $C_j$. Two component vertices are adjacent in $\mathcal{C}_\theta$ if there exists at least one healthy Gaussian graph edge joining the corresponding components.

Fig.~\ref{fig:component-graph} shows the component-level reduction: once the fault-pruned tree is contracted into \(c\) healthy components, any non-redundant repair must connect those components using at least \(c-1\) external crossing edges.

\begin{figure}[H]
\centering
\includegraphics[width=0.95\linewidth]{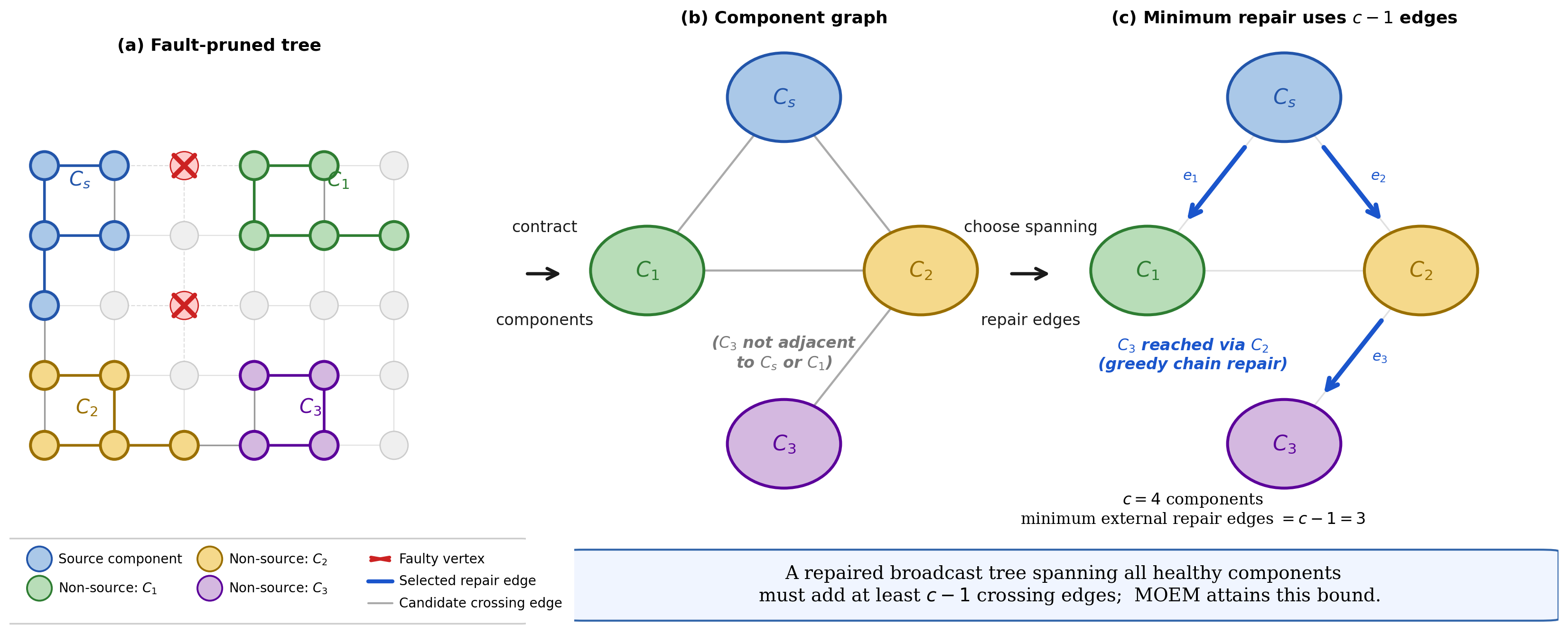}
\caption{Component contraction view of MOEM: fault-pruned tree components, the induced component graph, and the minimum \(c-1\) repair edges.}
\label{fig:component-graph}
\end{figure}

The repair problem at the component level is to connect $C_s$ to all other components using component-crossing graph edges. Every selected crossing edge becomes an external repair edge.

\subsection{Edge-Minimum Component Repair}

At each step, MOEM chooses a crossing edge $(a,b)$ where $a$ is already in the connected repaired part and $b$ is in an unrepaired component $C$. The selected edge minimizes the predicted depth
\begin{equation}
    D(a,b,C)=\dist(s,a)+1+\operatorname{ecc}_C(b),
    \label{eq:preddepth}
\end{equation}
where $\operatorname{ecc}_C(b)$ is the maximum distance from $b$ to a node in component $C$ using the current internal tree of $C$. Ties prefer $H_2$-type repair edges, then smaller components, and then a deterministic edge order.

\begin{algorithm}[H]
\caption{Edge-Minimum Repair for One Orientation}
\label{alg:edge-min}
\begin{algorithmic}[1]
\REQUIRE Dense Gaussian network $\Gk$, source $s$, fault set $F$, orientation tree $T_\theta$
\ENSURE A repaired rooted tree over $\Gk-F$, or failure
\STATE Remove $F$ and incident tree edges from $T_\theta$
\STATE Compute components $C_1,\ldots,C_c$ and the source component $C_s$
\STATE Initialize the repaired component set $\mathcal{R}\leftarrow\{C_s\}$
\STATE Initialize the repaired edge set with the tree edges inside $C_s$
\WHILE{$|\mathcal{R}|<c$}
    \STATE Find all Gaussian edges from a component in $\mathcal{R}$ to a component outside $\mathcal{R}$
    \STATE Select the edge minimizing~\eqref{eq:preddepth}, with deterministic tie-breaks
    \STATE Add the selected crossing edge and the new component to the repaired tree
\ENDWHILE
\STATE Root the resulting edge set at $s$ and validate the tree
\end{algorithmic}
\end{algorithm}

\subsection{Multi-Orientation Selection}

The full MOEM algorithm evaluates the full constant-size family of eight orientations. The final candidate is selected lexicographically by
\[
    (\text{success},\ \text{depth},\ \text{repair edges},\ \text{maximum out-degree},\ \text{orientation rank}),
\]
where successful candidates are ranked before unsuccessful candidates.

\begin{algorithm}[H]
\caption{MOEM Broadcast Repair}
\label{alg:moem}
\begin{algorithmic}[1]
\REQUIRE $\Gk$, source $s$, fault set $F$, orientation family $\Theta$
\ENSURE Non-redundant broadcast tree over $\Gk-F$
\STATE Use the full orientation family, $\Theta_F\leftarrow\Theta$
\STATE $\mathcal{B}\leftarrow\emptyset$
\FORALL{$\theta\in\Theta_F$}
    \STATE Construct the fault-free oriented Gaussian tree $T_\theta$
    \STATE Apply Algorithm~\ref{alg:edge-min} to $T_\theta$
    \STATE Insert the validated candidate into $\mathcal{B}$
\ENDFOR
\STATE Return the best valid candidate in $\mathcal{B}$ under the lexicographic score
\end{algorithmic}
\end{algorithm}

\subsection{Correctness}

\begin{lemma}
Let $T$ be a spanning tree of $\Gk$, and let $F$ be a set of faulty vertices not containing the source. If $T-F$ has $c$ connected components, then any repaired broadcast tree that uses all healthy vertices must add at least $c-1$ component-crossing edges.
\end{lemma}

\begin{proof}
Contract each connected component of $T-F$ to one supernode. A broadcast tree spanning all healthy vertices must connect these $c$ supernodes. Any connected graph on $c$ vertices has at least $c-1$ edges. Each such edge corresponds to a healthy graph edge crossing between two components. Thus at least $c-1$ component-crossing edges are necessary.
\end{proof}

\begin{theorem}
\label{thm:minrepair}
Suppose the component graph $\mathcal{C}_\theta$ of $T_\theta-F$ is connected. Then the edge-minimum repair step returns a non-redundant broadcast tree over all healthy vertices and uses exactly $c-1$ external component-repair edges.
\end{theorem}

\begin{proof}
Because $\mathcal{C}_\theta$ is connected, there exists a spanning tree of the component graph rooted at the source component. The repair procedure selects one crossing edge for each non-source component, so it adds exactly $c-1$ external edges. Inside each component, the original fault-pruned structure is a tree. Contracting each component after the selected crossings are added gives exactly the chosen component-level tree; therefore any cycle in the repaired graph would contract to a cycle in that component tree, which is impossible. Hence connecting these component trees cannot create a cycle and connects all healthy vertices. Rooting the resulting connected acyclic graph at $s$ gives exactly one parent to every healthy non-source node and no parent to $s$. Therefore the resulting structure is a non-redundant broadcast tree.
\end{proof}

\begin{corollary}
MOEM returns a correct non-redundant broadcast whenever at least one tested orientation has a connected healthy component graph.
\end{corollary}

The experiments in Section~\ref{sec:validation} found this condition in all tested one- and two-fault cases.

\begin{theorem}[Orientation dominance]
Let $S(\cdot)$ be the lexicographic candidate score used by MOEM. Since MOEM evaluates the full orientation family $\Theta$, if $T_\theta$ is any fixed orientation in $\Theta$, then the candidate returned by MOEM has score no worse than the repaired candidate obtained from $T_\theta$ alone.
\end{theorem}

\begin{proof}
MOEM evaluates every orientation in $\Theta$ and returns the valid candidate with minimum score. Since the repaired candidate of any fixed $T_\theta\in\Theta$ is one of the candidates considered, the selected candidate cannot have a worse score. The result follows directly from the minimization step.
\end{proof}

\section{The $k+2$ Depth Theorem}
\label{sec:proof-depth}

The previous section proves that, after an orientation has been selected, the component-repair phase is correct and uses the minimum possible number of external component-crossing repair edges. This section proves the depth bound used by MOEM. The proof has three ingredients. First, a certificate lemma converts local attachment inequalities into a global depth bound. Second, a finite set of coordinate repair templates covers one-fault, same-axis, opposite-axis, and separated off-axis two-fault placements. Third, the formerly hard orthogonal-axis pattern is handled by an explicit O6 five-component decomposition.

\begin{definition}[Repair certificate]
Fix an orientation $\theta$ and let $T_\theta-F$ have components $C_1,\ldots,C_c$, with source component $C_s$. A $K$-depth repair certificate is an ordering of the non-source components and crossing edges
\[
    (a_j,b_j),\qquad j=1,\ldots,c-1,
\]
where $a_j$ lies in the union of the source component and previously attached components, $b_j\in C_j$, and
\begin{equation}
    d_j(a_j)+1+\ecc_{C_j}(b_j)\le K .
    \label{eq:certificate}
\end{equation}
Here $d_j(a_j)$ is the depth of $a_j$ in the partially repaired tree before $C_j$ is attached, and $\ecc_{C_j}(b_j)$ is measured inside the original fault-pruned tree component $C_j$.
\end{definition}

\begin{lemma}[Certificate implies bounded depth]
\label{lem:certificate-depth}
If an orientation $\theta$ admits a $K$-depth repair certificate, then there exists a non-redundant repaired broadcast tree of depth at most $K$.
\end{lemma}

\begin{proof}
Attach the components in the order specified by the certificate. When component $C_j$ is attached through $(a_j,b_j)$, every vertex $v\in C_j$ obtains a path from the source through $a_j$, then the crossing edge $(a_j,b_j)$, then the unique internal tree path in $C_j$ from $b_j$ to $v$. Therefore
\[
    \dist(s,v)\le d_j(a_j)+1+\dist_{C_j}(b_j,v)
           \le d_j(a_j)+1+\ecc_{C_j}(b_j)
           \le K.
\]
Previously attached components already satisfy the same bound by induction. Each step connects a new tree component to the already repaired tree by one crossing edge, so no cycle is created and every healthy vertex is included exactly once. Hence the repaired structure is non-redundant and has depth at most $K$.
\end{proof}

\begin{table}[H]
\centering
\caption{Coordinate repair templates used in the $k+2$ proof.}
\label{tab:repair-templates}
\scriptsize
\begin{adjustbox}{max width=\textwidth}
\begin{tabular}{p{0.19\linewidth} p{0.25\linewidth} p{0.25\linewidth} p{0.20\linewidth}}
\toprule
Fault signature & Orientation choice & Detached component shape & Certificate bound\\
\midrule
One off-axis fault $(u,v)$, $uv\ne0$ & Reduce the smaller absolute coordinate first; ties use the transverse sign preference & One monotone triangular or trapezoidal suffix with a lateral boundary edge & $(\rho(f)+1)+1+(k-\rho(f))\le k+2$\\
One axis fault $(a,0)$ & Use a transverse-first orientation near the cut & One axis interval or one-row cap & $(a+1)+1+(k-a)\le k+2$\\
Two faults on the same ray & Sign preference keeps the source-side interval intact & One or two nested axis intervals & Attach the outer interval first; each value is at most $k+2$\\
Two faults on opposite rays & Use the sign-preference orientation that separates the tails & Two independent axis intervals on opposite sides & Each interval has a source-side lateral entry and remaining length at most $k-r$\\
Separated off-axis O3 pair & Four-case split: Case A ($y_1\ne y_2$); Case B ($y_1=y_2$, same sign, ancestor); Cases C--D reduce to Case A & Two independent chains, an inner/outer same-row pair, or one/two-node boundary microcomponents & Cases B--D have value $\le k$; Case A has value $k$, except boundary wrap-around microcomponents with value $k+1\le k+2$\\
Orthogonal-axis pair $(a,0),(0,b)$ & Use the O6 orientation with lower half-ball in the source component & Axis tail, two row fragments, and upper cap & Explicit four-edge certificate in Lemma~\ref{lem:orthogonal-axis}\\
\bottomrule
\end{tabular}
\end{adjustbox}
\end{table}

\begin{lemma}[Greedy preserves an ordered depth certificate]
\label{lem:greedy-preservation}
Suppose the non-source components of $T_\theta-F$ admit an ordering such that, before each component in the ordering is attached, there exists at least one available crossing edge with predicted depth value at most $K$. Then Algorithm~\ref{alg:edge-min}, run on this orientation, produces a candidate of depth at most $K$.
\end{lemma}

\begin{proof}
At any iteration, the hypothesis supplies an available edge with predicted value at most $K$. Algorithm~\ref{alg:edge-min} selects a crossing edge minimizing the predicted value in~\eqref{eq:preddepth}; therefore the edge it actually selects also has value at most $K$. After this selection, the repaired side contains at least the components that would have been available in the certificate ordering at the same stage, so the next certified edge remains available when its component is considered. Repeating this argument gives a sequence of selected edges satisfying~\eqref{eq:certificate}. Lemma~\ref{lem:certificate-depth} then gives depth at most $K$.
\end{proof}

In Lemmas~\ref{lem:onefault}--\ref{lem:orthogonal-axis}, the certificate edges are given in an explicit attachment order. Therefore the availability condition in Lemma~\ref{lem:greedy-preservation} is verified inside the case proofs rather than assumed separately in Theorem~\ref{thm:k2}.

\subsection{Coordinate Templates}

For a point \(u=(x,y)\), let \(\rho(u)=|x|+|y|\). A detached component is called an \(r\)-suffix if every vertex in it lies on layers at least \(r\) and every internal parent step reduces the layer by one until the cut vertex is reached. Table~\ref{tab:repair-templates} summarizes the finite coordinate templates used in the \(k+2\) proof. The following elementary observation is used repeatedly.

\begin{lemma}[Layer-suffix attachment]
\label{lem:suffix}
Let $C$ be a detached suffix whose entry vertex $b$ has minimum layer at least $r+1$, and suppose there is a healthy crossing edge $(a,b)$ with $a$ in the repaired part and $\dist(s,a)\le r+1$. If $\ecc_C(b)\le k-r$, then attaching $C$ through $(a,b)$ has certificate value at most $k+2$.
\end{lemma}

\begin{proof}
Substituting the three assumed inequalities gives
\[
    \dist(s,a)+1+\ecc_C(b)\le (r+1)+1+(k-r)=k+2.
\]
Thus the crossing edge satisfies~\eqref{eq:certificate} with $K=k+2$.
\end{proof}

\subsection{One-Fault Shielding}

\begin{lemma}[One-fault $k+2$ certificate]
\label{lem:onefault}
For every non-source fault $f\in \Gk$, at least one MOEM orientation admits a $(k+2)$-depth repair certificate for $T_\theta-\{f\}$.
\end{lemma}

\begin{proof}
By symmetry it is enough to consider $f=(u,v)$ with $u\ge0$ and $v\ge0$; the other quadrants are obtained by reflecting signs and using the corresponding sign-preference orientation. Let $r=u+v$.

If $u,v>0$, choose the orientation that first reduces the smaller of $u$ and $v$. If $u=v$, choose the deterministic transverse tie orientation. In this orientation the vertices whose parent chain passes through $f$ form a monotone suffix of the quadrant rooted at the children of $f$. The suffix has remaining height at most $k-r$, because every internal parent step increases or decreases the layer by exactly one and no vertex of $B_k$ has layer larger than $k$. The transverse child of a neighbor of $f$ gives a lateral crossing edge $(a,b)$ from the source component into the suffix. The attaching endpoint $a$ lies on layer at most $r+1$. Therefore Lemma~\ref{lem:suffix} gives a certificate value at most $k+2$.

If $f=(a,0)$ is on an axis, choose an orientation that avoids making the whole adjacent half-ball depend on the positive $x$-axis when a transverse reduction is available. The detached set is then either the interval $\{(x,0):a+1\le x\le k\}$ or a one-row cap adjacent to it. The interval entry is adjacent to a vertex of depth at most $a+1$, and the farthest endpoint is at internal distance at most $k-a$. Hence the certificate value is at most $(a+1)+1+(k-a)=k+2$. The same calculation applies to the one-row cap, whose internal eccentricity is no larger than the remaining radius to the boundary. Thus every one-fault placement has a $(k+2)$ certificate.
\end{proof}

\subsection{Two-Fault Generic Shielding}

\begin{lemma}[Same-axis and opposite-axis two-fault certificates]
\label{lem:axis-pairs}
Let $F=\{f_1,f_2\}$ with both faults on one coordinate axis, or on two opposite coordinate rays. Then some MOEM orientation admits a $(k+2)$-depth repair certificate for $T_\theta-F$.
\end{lemma}

\begin{proof}
Consider first two faults on the same ray, say $(a,0)$ and $(b,0)$ with $1\le a<b\le k$. Use the sign-preference orientation whose source component contains the layer immediately before the first fault. Removing the two faults can detach only axis intervals: the finite interval between the faults and the tail beyond the second fault, with empty intervals ignored. The interval between the faults is entered from a source-side lateral neighbor at depth at most $a+1$ and has length at most $b-a-1\le k-a-1$. Its certificate value is therefore at most $(a+1)+1+(k-a-1)=k+1$. The tail beyond $b$ is entered from a lateral neighbor of depth at most $b+1$ and has remaining length at most $k-b$, giving at most $k+2$.

If the faults lie on opposite rays, for example $(a,0)$ and $(-b,0)$, choose the sign-preference orientation that keeps the two tails on different sides of the source component. Each detached tail has a lateral entry adjacent to the source-side segment. For the positive tail the same bound as above gives $(a+1)+1+(k-a)\le k+2$; for the negative tail the reflected calculation gives $(b+1)+1+(k-b)\le k+2$. The $y$-axis cases are coordinate-swapped copies. Hence a certificate exists in all same-axis and opposite-axis cases.
\end{proof}

\begin{lemma}[Separated off-axis two-fault certificate]
\label{lem:offaxis}
Let $F=\{f_1,f_2\}$ be a two-fault placement in which both faults are off-axis, and assume that the pair is not reducible to a same-axis, opposite-axis, or orthogonal-axis pair by sign reflection and coordinate exchange. Then at least one MOEM orientation admits a $(k+2)$-depth repair certificate for $T_\theta-F$. More strongly, every separated off-axis component has certificate value at most $k$, except for boundary wrap-around microcomponents whose value is $k+1$.
\end{lemma}

\begin{proof}
Write $f_i=(x_i,y_i)$ and $r_i=\rho(f_i)=|x_i|+|y_i|$. A normalized representative means the image of the fault pair after applying sign reflections and, if necessary, exchanging the two coordinates, so that the detached chain being analyzed is entered from the positive-$x$ side and $y_i\ne0$. The relative position of two off-axis faults falls into the four cases below.
\begin{center}
\scriptsize
\begin{tabular}{p{0.12\linewidth}p{0.35\linewidth}p{0.38\linewidth}}
\toprule
Case & Relative position & Certificate value\\
\midrule
A & $y_1\ne y_2$ & two independent chains; value $k$, or $k+1$ for boundary wrap-around microcomponents\\
B & $y_1=y_2$, same sign, ancestor & inner component $<k$, outer component $k$\\
C--D & remaining same-row cases & reduce by reflection or priority exchange to Case A\\
\bottomrule
\end{tabular}
\end{center}

\emph{Case A: different rows, $y_1\ne y_2$.} The two detached components are two independent monotone chains. For the chain generated by $f_i$, use the crossing whose repaired-side endpoint is
\[
    a_i=(x_i+1,\,y_i-\sgn(y_i)),
\]
and whose component-side endpoint $b_i$ is the adjacent first vertex of the detached chain. This is a unit Gaussian edge in the coordinate model. The endpoint $a_i$ lies in the source component because the row shift $y_i\mapsto y_i-\sgn(y_i)$ moves to the transverse side of the chain, and the two faults lie on different rows, so they cannot both block this side entry. In the normalized representative,
\[
    \rho(a_i)=|x_i+1|+|y_i-\sgn(y_i)|=|x_i|+|y_i|=r_i .
\]
The remaining internal chain length from $b_i$ is
\[
    \ecc(b_i)=k-r_i-1.
\]
Therefore the certificate value is
\[
    \dist(s,a_i)+1+\ecc(b_i)
       = r_i+1+(k-r_i-1)=k .
\]
The only boundary subcase occurs when the detached component is a single boundary node or a two-node boundary component. Let $b$ be the boundary entry vertex of the detached component and let $\boldsymbol e\in\{(1,0),(-1,0),(0,1),(0,-1)\}$ be the unit step from $b$ toward the repaired side in the infinite Gaussian grid. If $b+\boldsymbol e\in B_k$, this is the ordinary side-entry crossing. Otherwise choose the unique quotient-lattice shift
\[
    \lambda=m(k,k+1)+n(-(k+1),k)
\]
such that
\[
    a=b+\boldsymbol e-\lambda\in B_k .
\]
The vertices $a$ and $b+\boldsymbol e$ represent the same residue class modulo the Gaussian ideal generated by $k+(k+1)i$, so $\phi(a)\equiv\phi(b+\boldsymbol e)\pmod N$ and $\{a,b\}$ is a valid Gaussian generator edge in $G_k$. For the one-node boundary component, $\ecc_C(b)=0$; for the two-node boundary component, $\ecc_C(b)=1$. The wrapped repaired-side endpoint is reached one layer later than the ordinary side entry, so the corresponding certificate value is at most $k+1\le k+2$.

\emph{Case B: same row, same sign, ancestor relation.} Suppose $y_1=y_2$ and the two faults have the same row sign. If one fault is an ancestor of the other in the selected row orientation, the deletion creates an inner component and an outer component. Attach the inner component first using the Case-A crossing at the inner fault $f_1$ after the local coordinate swap that exposes the transverse side row. Since the inner component terminates before the outer boundary layer, its certificate value is strictly less than $k$. Then attach the outer component through the layer-eccentricity crossing at the outer fault $f_2$. The repaired-side endpoint has depth equal to the cut layer, and the remaining eccentricity is exactly the remaining distance to the boundary. Hence the outer certificate value is $k$.

\emph{Cases C and D: same-row reductions.} The remaining same-row configurations are the non-ancestor and opposite-side same-row configurations. In both cases, a reflection of the row or an exchange of the two coordinate priorities separates the two damaged chains by a healthy transverse side entry. After that reduction, each detached component is exactly the independent-chain situation of Case A. Thus each such component has certificate value at most $k$, with only the same one- or two-node boundary wrap-around exception of value $k+1$.

For two off-axis faults, either the row coordinates differ, or they are equal; in the equal-row case the faults are either in an ancestor relation, a non-ancestor relation under the selected priority, or lie in the opposite-side same-row configuration. Thus the four cases are exhaustive. Therefore every separated off-axis two-fault placement admits a repair certificate of value at most $k+1$, and hence at most $k+2$.
\end{proof}

\subsection{Orthogonal-Axis Case}

The exhaustive proof-mining pass identified the orthogonal-axis pattern as the only case in which a single obvious priority tree may have only one good orientation. The O6 component inspection closes this case. The statement below is written for the positive pair; the other three sign patterns follow by reflection, and exchanging the two axes gives the coordinate-swapped case.

\begin{figure}[H]
\centering
\includegraphics[width=0.95\linewidth]{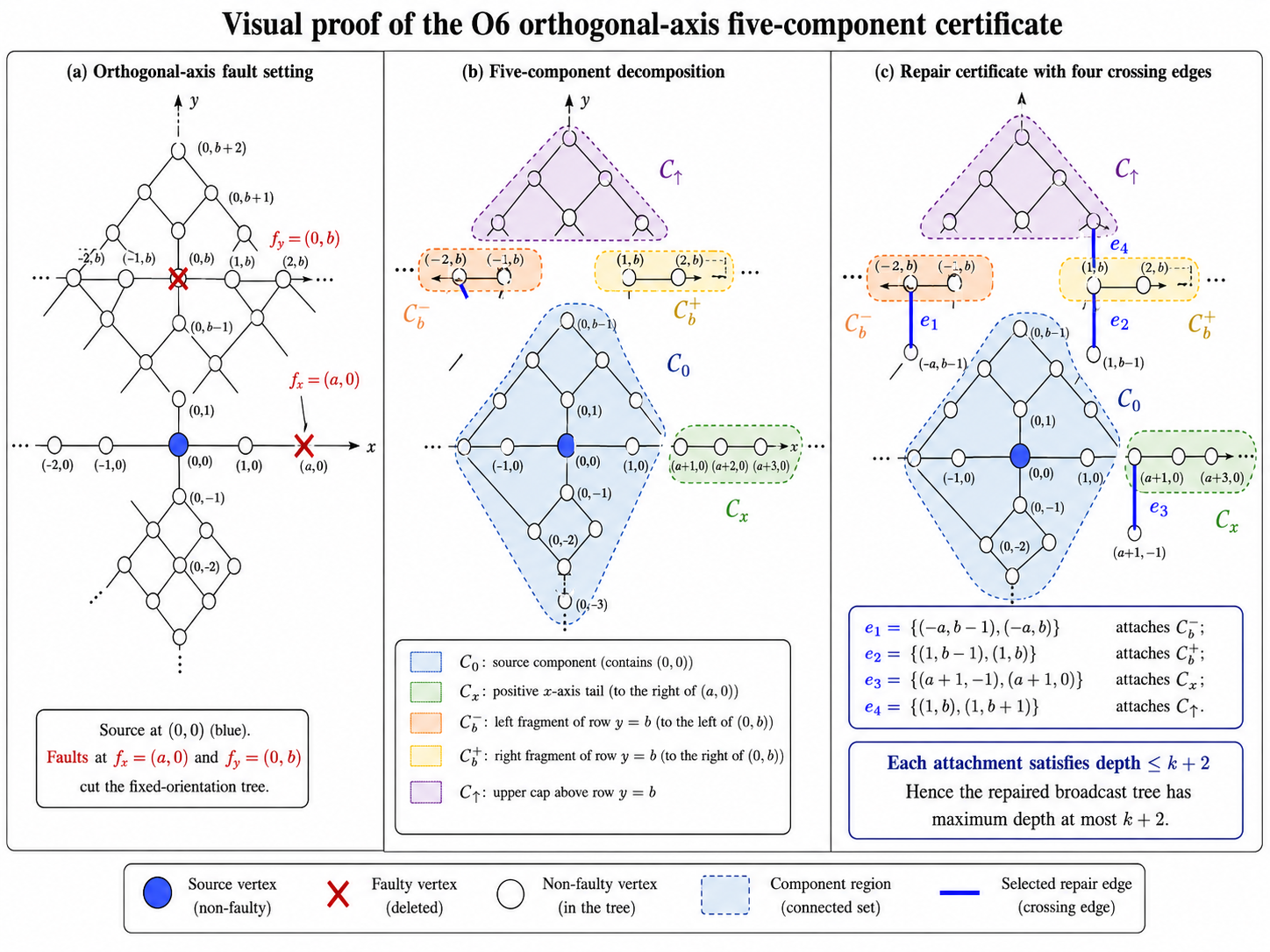}
\caption{O6 orthogonal-axis certificate: five fault-pruned components and four repair edges achieving depth at most \(k+2\).}
\label{fig:o6-components}
\end{figure}

\begin{lemma}[Orthogonal-axis two-fault certificate]
\label{lem:orthogonal-axis}
Assume, up to symmetry, that
\[
    F=\{(a,0),(0,b)\},\qquad 1\le a,b,\qquad a+b\le k.
\]
Then there is a MOEM orientation whose fault-pruned tree has a $(k+2)$-depth repair certificate.
\end{lemma}

\begin{proof}
Use the sign-symmetric orientation whose source component contains the lower half-ball and whose positive-axis faults cut only forward suffixes. The O6 component inspection shows that removing $(a,0)$ and $(0,b)$ produces exactly five components. The source component is denoted by $C_0$. The four detached components are
\[
    C_x=\{(x,0):a+1\le x\le k\},
\]
\[
    C_b^- =\{(x,b):-(k-b)\le x\le -1\},
\]
\[
    C_b^+ =\{(x,b):1\le x\le k-b\},
\]
and
\[
    C_\uparrow=\{(x,y):b+1\le y\le k,\ |x|+y\le k\}.
\]
Thus the two faults do not create arbitrary fragmentation; they create one positive $x$-axis tail, two row fragments on row $y=b$, and one upper cap. Fig.~\ref{fig:o6-components} visualizes this five-component decomposition and the four certificate edges used in the proof.

Attach the row fragments first. The left row fragment $C_b^-$ is attached by
\[
    e_1=\{(-a,b-1),(-a,b)\},
\]
where $(-a,b-1)\in C_0$ and $(-a,b)\in C_b^-$. The farthest vertex of $C_b^-$ from $(-a,b)$ is $(-(k-b),b)$, so
\[
    \ecc_{C_b^-}(-a,b)=(k-b)-a.
\]
Since $(-a,b-1)$ has depth $a+b-1$, the certificate value is
\[
    (a+b-1)+1+((k-b)-a)=k.
\]
The right row fragment $C_b^+$ is attached by
\[
    e_2=\{(1,b-1),(1,b)\}.
\]
The attaching endpoint $(1,b-1)$ has depth $b$, and the farthest vertex of $C_b^+$ from $(1,b)$ is $(k-b,b)$, at internal distance $k-b-1$. Hence
\[
    b+1+(k-b-1)=k.
\]
Next attach the positive $x$-axis tail by
\[
    e_3=\{(a+1,-1),(a+1,0)\}.
\]
The endpoint $(a+1,-1)$ lies in $C_0$ and has depth $a+2$. The farthest vertex of $C_x$ from $(a+1,0)$ is $(k,0)$, at internal distance $k-a-1$. Therefore
\[
    (a+2)+1+(k-a-1)=k+2.
\]
Finally attach the upper cap through the already repaired right row fragment by
\[
    e_4=\{(1,b),(1,b+1)\}.
\]
After $C_b^+$ is attached, the vertex $(1,b)$ has depth $b+1$. The entry point $(1,b+1)$ reaches every vertex of $C_\uparrow$ within at most $k-b$ internal tree steps, because the cap has remaining height $k-b$ and each parent step inside the cap decreases the Manhattan layer by one. Hence the cap certificate value is
\[
    (b+1)+1+(k-b)=k+2.
\]
All four detached components therefore admit attachments of value at most $k+2$. By Lemma~\ref{lem:certificate-depth}, the orthogonal-axis case has a non-redundant repaired tree of depth at most $k+2$.
\end{proof}

\begin{theorem}[MOEM $k+2$ depth theorem for one and two faults]
\label{thm:k2}
For every $k\ge5$ and every fault set $F\subseteq V(\Gk)$ with $|F|\le2$ and source not faulty, the MOEM orientation family contains an orientation whose edge-minimum repair has depth at most $k+2$. Consequently, MOEM returns a valid non-redundant broadcast tree of depth at most $k+2$.
\end{theorem}

\begin{proof}
If $|F|=0$, the fault-free oriented broadcast tree has depth at most $k$. If $|F|=1$, Lemma~\ref{lem:onefault} gives an ordered $(k+2)$ certificate. If $|F|=2$, then the pair is one of the following: same-axis or opposite-axis, separated off-axis, or orthogonal-axis after sign reflection and coordinate exchange. These cases are covered by Lemmas~\ref{lem:axis-pairs}, \ref{lem:offaxis}, and~\ref{lem:orthogonal-axis}, respectively. In each case, the proof gives both the crossing edges and the order in which the corresponding components are attached, so the hypothesis of Lemma~\ref{lem:greedy-preservation} is satisfied for that orientation. Hence the edge-minimum repair candidate for the certified orientation has depth at most $k+2$. MOEM selects the valid candidate of minimum lexicographic score, with depth preceding repair count among successful candidates. Therefore its selected candidate has depth no larger than this certified candidate. Correctness and non-redundancy follow from Theorem~\ref{thm:minrepair}, and the external repair-edge count for the selected orientation is minimum by Theorem~\ref{thm:minrepair}.
\end{proof}

\begin{corollary}[Minimum repair count under the depth bound]
For every one- or two-fault placement covered by Theorem~\ref{thm:k2}, the returned MOEM tree simultaneously satisfies the depth bound $k+2$ and uses exactly $c-1$ external component-repair edges for its selected fault-pruned orientation, where $c$ is the number of components of that orientation after deleting the faults.
\end{corollary}

\begin{proof}
The depth bound follows from Theorem~\ref{thm:k2}. The exact count $c-1$ follows from Theorem~\ref{thm:minrepair}. The minimum necessity follows from the component lower bound lemma.
\end{proof}

\section{Experimental Validation}
\label{sec:validation}

The proof is supported by two validation layers. First, exhaustive proof mining covers all one- and two-fault placements for $k=5,\ldots,10$. Second, large-scale validation covers $k=10,25,50,100,200$ under random and structured placements. Each reported case checks coverage, parent uniqueness, acyclicity, source reachability, fault exclusion, repair-edge count, and final depth.

Table~\ref{tab:exhaustive-by-k-v9} gives the condensed exhaustive results split by fault count. For each tested $k$, both fault counts have 100\% success, zero validation failures, maximum overhead two, and maximum repair count at most four. The six two-fault cases with exactly one good orientation are orthogonal-axis signatures; they require the balanced-smaller-first orientation within the tested family and attain depth $k+2$.

\begin{table}[H]
\centering
\caption{Condensed exhaustive validation split by fault count.}
\label{tab:exhaustive-by-k-v9}
\scriptsize
\begin{tabular}{c r c r c c c}
\toprule
$k$ & $N$ & $|F|$ & Cases & Max depth & Max repair & Min good\\
\midrule
5 & 61 & 1 & 60 & 7 & 1 & 6 \\
5 & 61 & 2 & 1,770 & 7 & 4 & 1 \\
6 & 85 & 1 & 84 & 8 & 1 & 6 \\
6 & 85 & 2 & 3,486 & 8 & 4 & 1 \\
7 & 113 & 1 & 112 & 9 & 1 & 6 \\
7 & 113 & 2 & 6,216 & 9 & 4 & 1 \\
8 & 145 & 1 & 144 & 10 & 1 & 6 \\
8 & 145 & 2 & 10,296 & 10 & 4 & 1 \\
9 & 181 & 1 & 180 & 11 & 1 & 6 \\
9 & 181 & 2 & 16,110 & 11 & 4 & 1 \\
10 & 221 & 1 & 220 & 12 & 1 & 6 \\
10 & 221 & 2 & 24,090 & 12 & 4 & 1 \\
\bottomrule
\end{tabular}
\end{table}

Table~\ref{tab:good-orientation-distribution} gives the full distribution of depth-safe orientations in the exhaustive run. It shows why orientation selection is essential: most placements are easy, but the two-fault regime contains six cases with exactly one good orientation. Those six cases are precisely the orthogonal-axis hard signatures handled by Lemma~\ref{lem:orthogonal-axis}.

\begin{table}[H]
\centering
\caption{Good-orientation counts in the exhaustive proof-mining run.}
\label{tab:good-orientation-distribution}
\scriptsize
\begin{tabular}{c r r r r r r r}
\toprule
$|F|$ & 1 & 2 & 4 & 5 & 6 & 7 & 8\\
\midrule
1 & 0 & 0 & 0 & 0 & 22 & 0 & 778 \\
2 & 6 & 12 & 190 & 332 & 2,891 & 39 & 58,498 \\
\bottomrule
\end{tabular}
\end{table}

Table~\ref{tab:smart-random} reports the large-scale random-fault validation with 1,000 trials per row. For $k\ge25$, MOEM succeeds in all trials, uses no fallback, keeps maximum depth at $k+2$, and has average repair count close to $|F|$.

\begin{table}[H]
\centering
\caption{Large-scale random-fault validation of MOEM.}
\label{tab:smart-random}
\scriptsize
\begin{adjustbox}{max width=\textwidth}
\begin{tabular}{c r c r c c c c c c c}
\toprule
$k$ & $N$ & $|F|$ & Trials & Succ. & Fallback & Avg. comp. & Avg. repair & Avg. depth & Max depth & Avg. ms\\
\midrule
10 & 221 & 1 & 1,000 & 100\% & 15.7\% & 2.007 & 2.357 & 10.167 & 11 & 5.6 \\
10 & 221 & 2 & 1,000 & 100\% & 26.0\% & 2.897 & 4.363 & 10.303 & 11 & 6.8 \\
25 & 1,301 & 1 & 1,000 & 100\% & 0.0\% & 1.943 & 0.943 & 25.077 & 27 & 49.0 \\
25 & 1,301 & 2 & 1,000 & 100\% & 0.0\% & 2.943 & 1.943 & 25.240 & 27 & 53.1 \\
50 & 5,101 & 1 & 1,000 & 100\% & 0.0\% & 2.003 & 1.003 & 50.073 & 52 & 338.6 \\
50 & 5,101 & 2 & 1,000 & 100\% & 0.0\% & 2.977 & 1.977 & 50.110 & 52 & 352.7 \\
100 & 20,201 & 1 & 1,000 & 100\% & 0.0\% & 1.997 & 0.997 & 100.013 & 102 & 2058.8 \\
100 & 20,201 & 2 & 1,000 & 100\% & 0.0\% & 2.990 & 1.990 & 100.063 & 102 & 2711.6 \\
200 & 80,401 & 1 & 1,000 & 100\% & 0.0\% & 1.993 & 0.993 & 200.013 & 202 & 10936.3 \\
200 & 80,401 & 2 & 1,000 & 100\% & 0.0\% & 2.980 & 1.980 & 200.027 & 202 & 15531.7 \\
\bottomrule
\end{tabular}
\end{adjustbox}
\end{table}

\begin{table}[H]
\centering
\caption{Deterministic structured MOEM stress tests at $k=200$, $|F|=2$.}
\label{tab:structured-k200}
\scriptsize
\begin{tabular}{l c c c c c}
\toprule
Mode & Succ. & Fall. & Comp. & Repair & Max depth\\
\midrule
axis trunk & 100\% & 0\% & 3.0 & 2.0 & 202\\
boundary leaf & 100\% & 0\% & 1.0 & 0.0 & 200\\
branch internal & 100\% & 0\% & 2.0 & 1.0 & 202\\
$H_1$ close & 100\% & 0\% & 2.0 & 1.0 & 202\\
$H_2$ close & 100\% & 0\% & 2.0 & 1.0 & 202\\
near boundary & 100\% & 0\% & 3.0 & 2.0 & 201\\
source adjacent & 100\% & 0\% & 4.0 & 3.0 & 202\\
two branches & 100\% & 0\% & 3.0 & 2.0 & 202\\
\bottomrule
\end{tabular}
\end{table}

Figure~\ref{fig:repair-count} summarizes the repair-count contrast. Fixed-orientation local repair can require a number of local adjustments that grows with $k$ under trunk-like or same-cycle stress, whereas MOEM's orientation selection keeps the corresponding close-generator repair count near $|F|$ and preserves the $k+2$ depth bound.

\begin{figure}[H]
\centering
\includegraphics[width=0.95\linewidth]{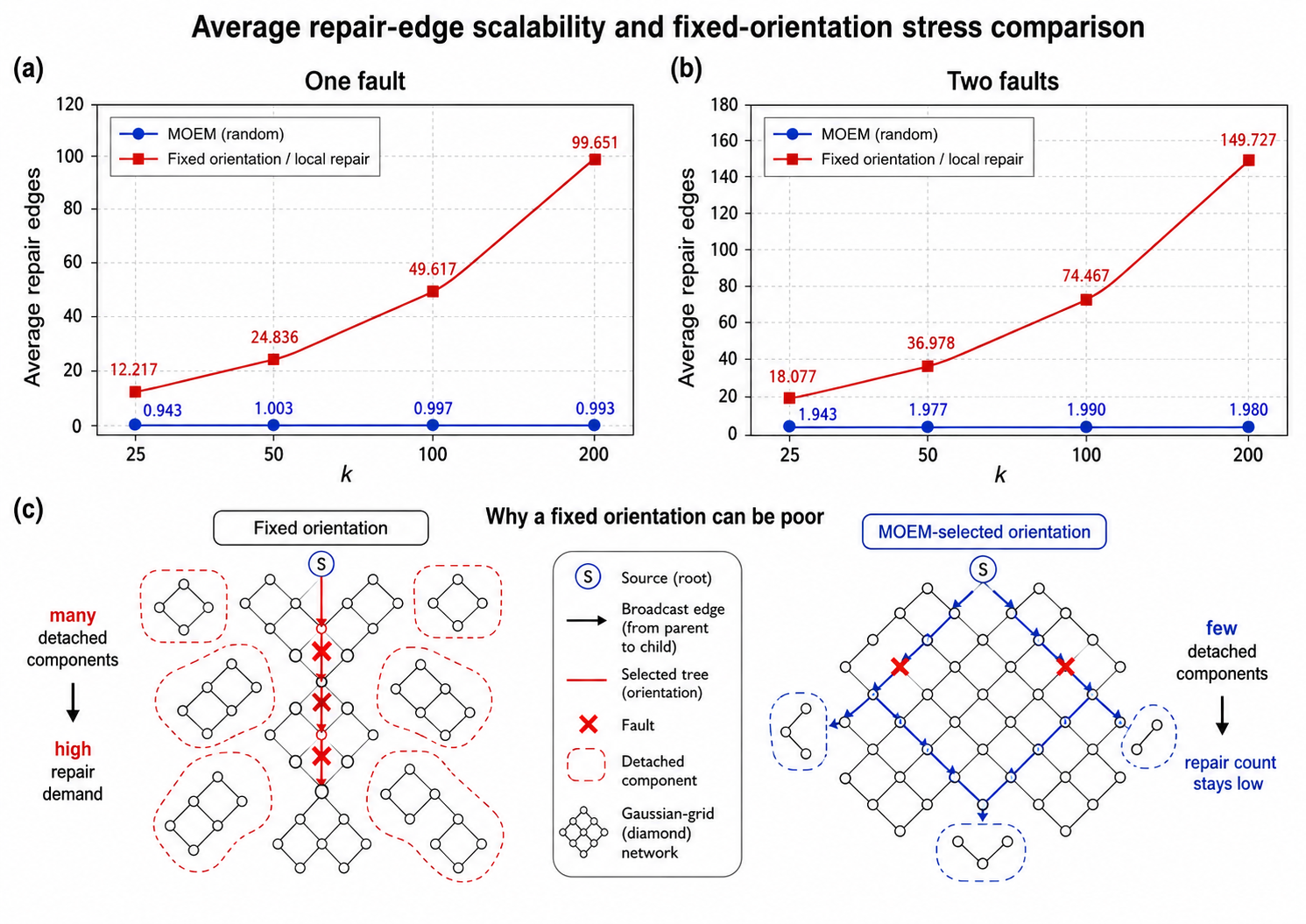}
\caption{Repair-edge scalability: MOEM stays near $|F|$, while fixed-orientation stress repair grows with $k$.}
\label{fig:repair-count}
\end{figure}

The hard O6 example is shown in Fig.~\ref{fig:o6-components}. For $k=10$ and $F=\{(3,0),(0,4)\}$, the fault-pruned tree has five components, four external repair edges, and final depth $12=k+2$. The simulation code and data used to generate the validation tables are available from the corresponding author upon reasonable request.

\section{Discussion}
\label{sec:discussion}

The proved guarantees are intentionally separated from empirical observations. Theorems~\ref{thm:minrepair} and~\ref{thm:k2} prove non-redundancy, $c-1$ external repair-edge optimality for the selected orientation, and depth at most $k+2$ for $|F|\le2$. The validation results quantify typical repair counts, identify rare hard signatures, and check the implementation at larger sizes.

The eight-orientation family is proved sufficient but not claimed minimal. The exhaustive data give a partial lower bound: the six one-good-orientation hard signatures are handled only by the balanced-smaller-first orientation. Hence any subfamily of the tested orientations that omits this orientation fails the $k+2$ theorem on those cases.

The main limitation is the restriction to $|F|\le2$. The component-repair theorem itself is more general: for any fixed fault set and any selected orientation whose healthy component graph is connected, MOEM repairs the components using the minimum possible $c-1$ external crossing edges. Since the oriented broadcast trees have degree at most four, deleting $q$ non-source vertices creates at most $1+3q$ components, so the component-repair cost remains $O(q)$. What remains open for $q\ge3$ is the depth-certificate theorem. Three collinear faults should be handled by nested axis-interval certificates; two off-axis faults plus one axis fault combine the O3 crossing with an axis-tail certificate unless the axis fault blocks the side entry. A natural conjecture is that a finite extension of the O3/O6 certificate library gives depth $k+O(1)$ for three faults while preserving $c-1$ repair optimality.

As a concrete partial three-fault certificate, consider three collinear faults on the positive $x$-axis, $F=\{(a,0),(b,0),(c,0)\}$ with $1\le a<b<c\le k$. In the transverse sign-preference orientation, the lower-side entries remain healthy, and the deleted vertices create at most two bounded axis intervals and one outer tail. A bounded interval cut after layer $p$ and ending before layer $q$ can be attached through $((p+1,-1),(p+1,0))$ with certificate value at most $(p+2)+1+(q-p-2)=q+1\le k+1$. The outer tail beyond $c$ is attached through $((c+1,-1),(c+1,0))$ with value $(c+2)+1+(k-c-1)=k+2$, with empty intervals ignored. Thus this representative three-fault family follows the same nested-interval logic as Lemma~\ref{lem:axis-pairs}; the open cases are those where non-collinear faults block the side entries used by such interval certificates.

Finally, the baseline choice is dictated by the objective. BFS rebuild and precomputed tree-diversity methods are useful resilience mechanisms, but they optimize connectivity or path diversity rather than the number of external component-crossing repair edges needed to preserve a damaged diameter-level tree. Fixed-orientation local repair is therefore the nearest same-objective baseline, and Fig.~\ref{fig:repair-count} shows the cost of omitting orientation selection.

\section{Complexity}

For each tested orientation, the fault-pruned forest can be computed in $O(N)$ time. Crossing edges can be found by scanning the four neighbors of each healthy vertex, also $O(N)$. For each tree component, all entry-vertex eccentricities can be computed in linear time in that component by the standard two-sweep tree-diameter method. Summed over all components, this preprocessing is $O(N)$ for one orientation. Since MOEM evaluates a constant number of orientations, the total running time remains $O(N)$ for the one- and two-fault setting, with a moderate constant factor.

\section{Conclusion}

This paper introduced MOEM, a multi-orientation edge-minimum repair method for non-redundant one-to-all broadcasting in dense Gaussian networks. MOEM evaluates a constant-size family of Gaussian broadcast-tree orientations, selects a fault-aware candidate, contracts the fault-pruned tree into healthy components, and reconnects those components using an externally edge-minimum set of repair edges. For a selected orientation with $c$ healthy components, exactly $c-1$ external component-crossing repair edges are necessary and sufficient.

For every one- or two-fault placement, the MOEM family contains an orientation admitting a repaired broadcast tree of depth at most $k+2$. The proof uses certificate preservation, coordinate shielding templates, an explicit four-case O3 proof for separated off-axis pairs, and a five-component O6 certificate for the orthogonal-axis hard case. Exhaustive validation over 62,768 cases for $k=5$ through $10$, together with large-scale validation through $k=200$, confirms the proof-derived behavior. The component-repair theorem extends to arbitrary fault counts when the healthy component graph is connected; the collinear three-fault example suggests how the certificate program can begin, while the main next step is to complete the depth-certificate classification for interacting non-collinear faults.

\section*{Acknowledgment}
The author thanks the Department of Computer Science, Faculty of Science, Kuwait University, for its support and research environment. This work did not receive a specific grant from any funding agency in the public, commercial, or not-for-profit sectors.

\end{document}